\documentclass[10pt,a4paper,english,twocolumn]{article}

\usepackage[utf8]{inputenc}
\usepackage{abstract}
\usepackage{algpseudocode}
\usepackage[section]{algorithm}
\usepackage{amsmath}
\usepackage{amsfonts}
\usepackage{amssymb}
\usepackage{amsthm}
\usepackage{authblk}
\usepackage{cancel}
\usepackage[font=small,labelfont=bf]{caption}
\usepackage{color}
\usepackage{csquotes}
\usepackage{enumitem}
\usepackage{float}
\usepackage[T1]{fontenc}
\usepackage{graphicx}
\usepackage{ifdraft}
\usepackage{marvosym}
\usepackage{orcidlink}
\usepackage{subfig}
\usepackage{titling}
\usepackage{varwidth}
\usepackage{xcite}
\usepackage{xcolor}
\usepackage{xpatch}
\usepackage{indentfirst}
\usepackage{lipsum}
\usepackage[a4paper, total={7in, 10in}]{geometry}
\allowdisplaybreaks
\usepackage{mathtools}

\usepackage{hyperref}
\hypersetup{
	colorlinks,
    citecolor=blue,
    filecolor=blue,
    linkcolor=black,
    urlcolor=black
}

\usepackage{textcomp}
\usepackage{scalerel}

\def\kiwi{\scalerel*{\includegraphics{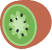}}{\textrm{\textbigcircle}}}

\usepackage[capitalise, nameinlink, noabbrev]{cleveref}
\crefname{prop}{Property}{Properties}
\Crefname{prop}{Property}{Properties}
\crefname{subs}{Subsection}{Subsections}
\Crefname{subs}{Subsection}{Subsections}
\crefalias{inequality}{equation}
\crefname{inequality}{Inequality}{Inequalities}
\Crefname{inequality}{Inequality}{Inequality}
\creflabelformat{inequality}{#2\textup{(#1)}#3}
\crefalias{summation}{equation}
\crefname{summation}{Summation}{Summations}
\Crefname{summation}{Summation}{Summations}
\creflabelformat{summation}{#2\textup{(#1)}#3}

\usepackage[backend=bibtex8, style=ieee, url=false, doi=false, isbn=false]{biblatex}
\newcommand{\inputtikz}[1]{%
		\includegraphics{tikz/compiled/#1.pdf}
}

\theoremstyle{plain}
\newtheorem{theorem}{Theorem}
\newtheorem*{theorem*}{Theorem}

\newtheorem*{lemma*}{Lemma}
\newtheorem{corollary}{Corollary}
\theoremstyle{definition}
\newtheorem{definition}{Definition}
\newtheorem{property}{Property}
\newtheorem{example}{Example}

\newcommand{\CodeAnnotation}[4]{
	\hspace*{#1}
	\rlap{
		\smash{
			\raisebox{\dimexpr#2\normalbaselineskip+#2\jot}{
$\left.
\begin{array}
	{@{}c@{}}
	#4
\end{array}
\color{red}
\right\}%
\color{red}
\begin{tabular}{l}
	\hspace*{-1em} \scriptsize #3
\end{tabular}$
			}
		}
	}
}

\newcommand{\floor}[1]{\left\lfloor #1 \right\rfloor}

\newcommand{\bunder}[2]{\underbrace{#1}_{\text{\scriptsize\tabular[t]{@{}c@{}}#2\endtabular}}}

\newcommand{\qmerge}{\oplus}
\newcommand{\qadd}{+}
\newcommand{\defeq}{\overset{\operatorname{def}}{=}}
\newcommand{\bigo}{\mathcal{O}}
\newcommand{\qlimit}{\floor{\frac{n}{k}}}
\newcommand{\compressCost}[1][m]{\ensuremath{\bigo({#1} \log \sigma)}}

\xpatchcmd{\algorithmic}{\setcounter}{\algorithmicfont\setcounter}{}{}
\providecommand{\algorithmicfont}{}
\providecommand{\setalgorithmicfont}[1]{\renewcommand{\algorithmicfont}{#1}}

\algdef{SE}[DOWHILE]{Do}{doWhile}{\algorithmicdo}[1]{\algorithmicwhile\ #1}%

\begin{document}

\title{Techniques for Authenticating Quantile Digests}
\date{}
\author{Alessandro Scala\thanks{a.scala2@studenti.unipi.it} \orcidlink{0009-0001-4739-0451}}
\affil{\textit{University of Pisa, Department of Computer Science}}

\twocolumn[
  \begin{@twocolumnfalse}
    \maketitle
    \begin{abstract}
		\noindent
		We investigate two possible techniques to authenticate the \emph{q-digest} data structure, along with a worst-case study of the computational complexity both in time and space of the proposed solutions, and considerations on the feasibility of the presented approaches in real-world scenarios. We conclude the discussion by presenting some considerations on the information complexity of the queries in the two proposed approaches, and by presenting some interesting ideas that could be the subject of future studies on the topic.
    \end{abstract}
	\vspace*{2em}
  \end{@twocolumnfalse}
]
\saythanks

\section*{Introduction}
With the proliferation of distributed networks and the increase in the amount of data that needs to be handled by nodes of these networks, and with the rise of networks where different nodes have different computing capabilities, there is a clear need to develop compact data structures that support efficient queries. However, where there is a public network, there is an opening for malicious users to manipulate data to their advantage. In this untrusted setting, we need to devise methods to ensure the integrity of data and the correctness of queries. We therefore investigate the possibility of authenticating q-digests -- a compact data structure introduced in \cite{shrivastava2004medians} -- and queries performed on them, in order to fulfil this need.

The main contribution of this paper is to provide a new theorem for the size bound of q-digests, new algorithms for q-digest compression together with their correctness proofs, and two methods for authenticating the data structure.

The paper is structured as follows: in \cref{sec:background} we recall the general details of q-digests and introduce some terminology and notation. In \cref{sec:compressionIssues} we present an issue related to the correctness of the compression procedure, along with ways to prevent it from happening. In \cref{sec:sizeBoundIssues} we discuss another issue, related to the theoretical bounds we have on the size of a q-digest, and prove a new theorem limiting the growth of a q-digest. We then introduce a first authentication method for q-digests in \cref{sec:wda} and discuss its limitations. This proposal is then improved on with a second authentication method introduced in \cref{sec:kvc-qa}. We discuss the benefits of using one method over the other in \cref{sec:wdavskvcqa}, and conclude in \cref{sec:furtherDevelopments} by discussing possible improvements and research topics for future works.

\section{Background}\label{sec:background}
The \emph{q-digest} data structure has been introduced in \cite{shrivastava2004medians} with the aim of succinctly representing a distribution of integer values in a range $[1,\sigma]$. It does so by collecting equal values inside buckets, and then merging adjacent buckets if the number of elements they contain is too small. This compression has both the benefit of reducing the space footprint of the structure, and of consequently reducing the complexity of queries, making it a very efficient and effective structure for sensor networks, where computational resources are limited.

\noindent The queries supported by the q-digest are:
\begin{itemize}
    \item \textbf{quantile query}: given a real number $q \in [0,1]$, return the smallest value in the distribution that is greater or equal to $qn$ values, where $n$ is the number of values stored in the digest. In other words, return the value that is in position $qn$ in the list obtained by sorting the values of the distribution in increasing order;
    \item \textbf{inverse quantile query}: given a value $x$, return its rank in the sequence sorted as before;
    \item \textbf{range query}: given a range $[l,r]$, return the number of values that fall into this range;
    \item \textbf{consensus query}: given a real number $s \in [0,1]$, return all the values that have multiplicity of at least $sn$.
\end{itemize}
In this paper, we focus on quantile queries, and consequently on range queries as they can be derived from quantile ones.

We now report a few definitions and conventions that will be used throughout this paper.

The \textbf{q-digest tree} of a q-digest $Q$ is the binary partition tree over the domain of $Q$. For simplicity, empty nodes can be interchangeably said to have $count = 0$ or to not have a $count$ associated.
Unless explicitly stated, we will use the term ``nodes'' to refer to nodes of the q-digest tree in general, and ``buckets'' to specifically mean nodes with $count > 0$. We univocally assign \textbf{indices} to all nodes of a q-digest tree, according to a \emph{breadth-first visit}, starting from the index 1 for the root.\\
When viewing q-digests as binary trees, we use $\mathit{b.left}$, $\mathit{b.right}$, $\mathit{b.parent}$ to mean respectively the left child of $b$, its right child, and its parent, and $\mathit{b.cnt}$ for its count.\\
When viewing q-digests as sets of pairs $(index, count)$, we use $Q[i]$ to denote the element of the set with $index = i$.\\
We use $\left|Q\right|$ to indicate the number of buckets of a q-digest $Q$. This is exactly the cardinality of $Q$ when considered as a set of buckets.

We now introduce a function useful for writing some definitions and theorems in a more compact manner:

\begin{definition}[$\nabla$ Function]\label{def:nablaRelation}\;\\
    Let $Q$ be a q-digest, and let $b$ be a node.
    \begin{equation*}
    \nabla_Q(b) \defeq
    \begin{cases}
        b.cnt \hfill \text{if $b$ is the root of the q-digest tree}\\
        b.cnt + b.parent.cnt + b.sibling.cnt \quad \quad \text{o.w.}
    \end{cases}
    \end{equation*}
\end{definition}

This function is different from the $\Delta$ relation used in~\cite{shrivastava2004medians}, as that one goes top-down, i.e.,
$$
    \Delta_b \defeq \mathit{b.cnt + b.left.cnt + b.right.cnt}\,,
$$
while the $\nabla$ function goes bottom-up, and this is reflected in the symbol used to represent it.

We recall the definition of q-digest given in~\cite{shrivastava2004medians}, rewritten using our $\nabla$: 
\begin{definition}[Q-Digest]\label{def:qdigest}
    \;\\
    A subset $Q$ of nodes of a tree is a \emph{q-digest} if and only if
    \begin{gather*}
        \forall b \in Q \setminus \left(\mathit{\left\lbrace Q.root \right\rbrace \cup Q.leaves } \right).
    \end{gather*}
    \vspace*{-2em}
    \begin{align}
        b.cnt &\leq \qlimit \label[prop]{eq:qdigestProperty1}\\
        \nabla_Q(b) &> \qlimit \label[prop]{eq:qdigestProperty2}\,,
    \end{align}
    where $n$ is the number of values stored in the \emph{q-digest}, and $k$ is a compression parameter chosen at the time of creation of the q-digest.\\
    The root has to satisfy \cref{eq:qdigestProperty1}, but can violate \cref{eq:qdigestProperty2}, while the leaves have to satisfy \cref{eq:qdigestProperty2}, but can violate \cref{eq:qdigestProperty1}. Note that empty nodes are not considered in this definition.
\end{definition}
The \emph{compression parameter} $k$ governs the merging of buckets with low count into bigger ones. Specifically, when $k = 1$, all the values will be compressed into the root of the tree, since \cref{eq:qdigestProperty2} can never be satisfied. On the other hand, when $k > n$, \cref{eq:qdigestProperty2} is always verified, and the q-digest degenerates into a list of frequencies in the leaf buckets.

\emph{Summing} two digests (with the same $k$ and $\sigma$) means essentially performing a union of the sets of buckets of the two digests, with the caveat that when two buckets have the same index, only one bucket will be present into the resulting q-digest, with $count$ equal to the sum of the counts of the two source buckets.
\begin{definition}[Q-Digest sum]\label{def:qsum}\;\\
    Let $Q_1$ and $Q_2$ be two q-digests with the same $k$ and $\sigma$. Then:
    \begin{gather*}
        Q_\qadd = Q_1 \qadd Q_2 \,\text{, where}\\
        \forall i\in[1,2\sigma-1].\; Q_\qadd[i] = Q_1[i] + Q_2[i]
    \end{gather*}
\end{definition}
The term $2\sigma-1$ in this definition represents the number of nodes of the full binary partition tree over $[1,\sigma]$. It is trivially verifiable that $n_{Q_1 \qadd Q_2} = n_{Q_1} + n_{Q_2}$.

Note that the sum of two q-digests is \emph{not necessarily a q-digest} itself, since some buckets may violate \cref{eq:qdigestProperty2} in \cref{def:qdigest}, while \cref{eq:qdigestProperty1} is preserved. Indeed, since
\begin{equation*}
	\left(b_{Q_1}.cnt \leq \floor{\frac{n_{Q_1}}{k}}\right) \hspace*{1em} \land \hspace*{1em} \left(b_{Q_2}.cnt \leq \floor{\frac{n_{Q_2}}{k}}\right)\,,
\end{equation*}
we have that
\begin{equation*}
	b_{Q_1}.cnt + b_{Q_2}.cnt \leq \floor{\frac{n_{Q_1} + n_{Q_2}}{k}}
\end{equation*}

To merge two q-digests then it is sufficient to \emph{sum} the two q-digests, and then \emph{compress} them.

We denote with $\qmerge$ the merging operation:
\begin{definition}[Q-Digest merge]\label{def:qmerge}\;\\
$$
	Q_1 \qmerge Q_2 \defeq \textproc{Compress}(Q_1 \qadd Q_2)
$$
\end{definition}
Since we have now recompressed the sum and restored the q-digest properties (in particular \cref{eq:qdigestProperty2}), a merge of two q-digests is itself a q-digest.

This operation is mainly used for two purposes: during the creation of a q-digest, and to combine two q-digests.

\section{Issues with the Original \textproc{Compress} Algorithm}\label{sec:compressionIssues}
The original formulation of the algorithm \textproc{Compress} reported in \cite{shrivastava2004medians} involves performing one single pass from the bottom to the top of the tree. In the case of the creation of a new q-digest, this algorithm does indeed correctly compress the buckets so that the q-digest properties are satisfied. However, in the case of the merge of two q-digests, \emph{in some instances} the procedure might not work as intended, and leave some buckets violating \cref{eq:qdigestProperty2}. This issue is mentioned in \cite{cormode2010methods, gakhov2019probabilistic}, but no solution has been published as of the time of writing, as far as we know. The source of this issue is the parent count included in the property check.

In this section, we will analyse this issue and provide a solution in the form of two new compression algorithms that do not present the issue.

To this aim, we first note that during the creation of a new q-digest with the usual algorithm defined in \cite{shrivastava2004medians}, the following invariant holds true at all times during each step of the algorithm:

\begin{theorem}[Construction Invariant]\label{the:constructionInvariant}\;\\
	If $Q$ is a q-digest constructed by compressing a predefined set of frequencies, then
	\begin{equation}\label{eq:constructionInvariant}
		\forall b \in Q.\;\;\mathit{b.left} \notin Q \land \mathit{b.right} \notin Q\,,
	\end{equation}
	where with \(\mathit{b.left}\) and \(\mathit{b.right}\) we denote respectively the left and right child of the bucket \(b\)
\end{theorem}
\begin{proof}
	(Sketch) If $Q$ has been constructed using the compression procedure in \cite{shrivastava2004medians}, then every bucket $b$ is either a leaf, meaning it has no children, or it is an inner node that has been added by compressing its two children, which are now empty.
	Because the \textproc{Compress} algorithm used during its creation iterates only once on all buckets, starting from the bottom, it is not possible that a new child is added after the creation of $b$.
\end{proof}
\begin{corollary}\label{cor:constructionInvariantCorollary}\;\\
	If the \emph{construction invariant} holds, then
	\begin{gather*}
		\forall\, \mathit{buckets}\; b \in Q.\;\;\mathit{b.parent} \notin Q
	\end{gather*}
\end{corollary}

If we can guarantee that this holds true, then we can guarantee that the \textproc{Compress} algorithm returns a properly compressed q-digest that satisfies \cref{def:qdigest}.

Unfortunately, not all q-digest operations preserve the invariant. In particular, it is not preserved by sum or compression of a generic digest, and consequently it is not preserved by the merge operation, as shown in the following example.

\begin{example}\;\\
Let us discuss an example where the compression problem occurs. Suppose we have two digests $Q_1$ and $Q_2$, shown in \cref{fig:wrongCompress}, with $k = 4, n_1 = 38, n_2 = 36$, obtained respectively from the frequency sets (i.e., sets of tuples $(value, multiplicity)$)
\begin{align*}
	S_1 &= \{(1,1), (2,2), (3,3), (4,4), (5,6), (6,6), (7,7), (8,9)\}\\
	S_2 &= \{(1,8), (2,7), (3,6), (4,5), (5,4), (6,3), (7,2), (8,1)\} \,.
\end{align*}
Compressing these sets, we obtain the following \mbox{q-digests} (reminding we note them as sets of tuples $(index, count)$):
\begin{align*}
	Q_1 &= \{(4, 3), (5, 7), (12, 6), (13, 6), (14, 7), (15, 9)\}\\
	Q_2 &=  \{(6, 7), (7, 3), (8, 8), (9, 7), (10, 6), (11, 5)\}\\
	Q_1 \qadd Q_2 &= \{(4, 3), (5, 7), (6, 7), (7, 3), (8, 8), (9, 7), (10, 6)\\
	&\quad\;\;\;(11, 5), (12, 6), (13, 6), (14, 7), (15, 9)\}\\
	Q_1 \qmerge Q_2 &= \{(1, 10), (4, 18), (5, 18),(12, 6), \\
	&\quad\;\;\;(13, 6), (14, 7), (15, 9)\}
\end{align*}

Now, this last object \emph{does not verify} \cref{eq:qdigestProperty2} as, for example, it is not verified for node $15$, in fact:
\begin{equation*}
	\nabla_{Q_1 \qmerge Q_2}(15) =
	9 + 7 + 0 =
	16
	\quad\ngtr\quad
	18 =
	\bigg\lfloor \frac{74}{4} \bigg\rfloor =
	\bigg\lfloor \frac{n}{k} \bigg\rfloor
\end{equation*}

We can quickly verify with the example we just provided that the resulting object is not a q-digest even by choosing a different value for \(k\). If we choose for example $k'=7$, so that
\begin{equation*}
	\nabla_{Q_1 \qmerge Q_2}(13) = 6 + 6 + 0 = 12
	\quad>\quad
	10 = \floor{\frac{74}{7}} = \floor{\frac{n}{k'}} \,,
\end{equation*}
we notice that now
\begin{equation*}	
	Q[4].cnt = 18
	\quad>\quad
	10 = \floor{\frac{n}{k'}} \,,
\end{equation*}
meaning that now \cref{eq:qdigestProperty1} is not verified anymore by the bucket with index $4$. In fact, \emph{there is no choice of $k$} that makes both properties satisfied for all buckets.\\

The reason why this compression problem happens is that, by proceeding bottom-up, a merge at an upper level might invalidate checks made at the level directly below. By looking at the sum of the two q-digests in \cref{fig:wrongCompress}, we can in fact see that the property would be satisfied before compression, as we would have 
\begin{equation*}
	\nabla_{Q_1 \qadd Q_2}(15) =
	9+7+3 =
	19
	\quad>\quad
	18 =
	\qlimit
\end{equation*}

However, once we proceed to the upper level, we check the property for the parent node, $7$, and find
\begin{equation*}
	\nabla_{Q_1 \qadd Q_2}(7) =
	3+7+0 =
	10
	\quad<\quad
	18 =
	\qlimit
\end{equation*}

Thus, node $7$ needs to be merged with node $6$ into their parent $3$. This step, though, makes it so that $\nabla_{Q_1 \qadd Q_2}(15)$ has now changed to $16$, as the parent's count has changed to 0, invalidating the check performed previously.
\end{example}

\begin{figure}
	\centering
	\resizebox*{0.45\textwidth}{!}{
	\inputtikz{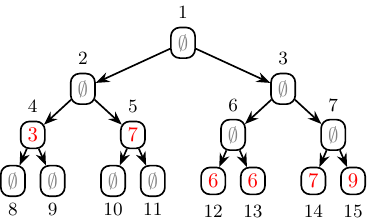}
}
\;\\
\;\\
\resizebox*{1cm}{!}{
	\begin{tikzpicture}
		\node[x=0, y=0, style={font=\Huge}] (A) {$\qadd$};
	\end{tikzpicture}
}
\;\\
\;\\
\resizebox*{0.45\textwidth}{!}{
	\inputtikz{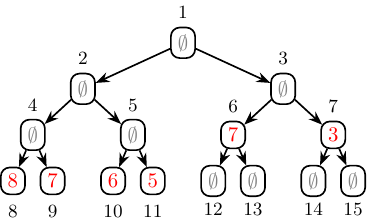}
}

\hfill
\;\\
\begin{tikzpicture}
	\node[x=0, y=0] (A) {};
	\draw[to-] (0.2,-0.7) -- (0.2,0.3);
\end{tikzpicture}
\;\\
\;\\

\resizebox*{0.45\textwidth}{!}{
	\inputtikz{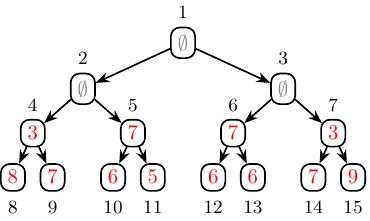}
}

\hfill
\;\\
\begin{tikzpicture}
	\node[x=0, y=0] (A) {};
	\draw[to-, transform canvas={xshift = 0.13cm}] (0.2,-0.7) -- (0.2,0.3) node[midway, fill=white, anchor=center] {\textproc{Compress}};
\end{tikzpicture}
\;\\
\;\\

\resizebox*{0.45\textwidth}{!}{
	\inputtikz{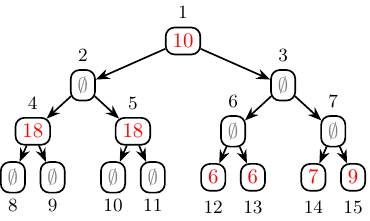}
}
	\caption{Compression during merging of two q-digests may lead to errors. Each node is labelled with its index next to it. The two starting digests have both $k=4$, and respectively $n=38$ and $n=36$. The resulting digest has $n=74$, and consequently $\qlimit = 18$.}
	\label{fig:wrongCompress}
\end{figure}

\paragraph{Consequences of the incorrectness of the compression algorithm}\;\\
\indent The object returned by this merging procedure when this problem happens is not strictly speaking a q-digest anymore, because, as we stated above, \cref{eq:qdigestProperty2} is violated for some buckets. Despite this violation, queries defined for regular q-digests still do work, and in fact, by restoring \cref{eq:qdigestProperty2}, the responses to queries on the q-digest would lose accuracy, as restoring the property involves compressing some buckets, causing loss of information.
However, something that has to be kept in mind are the consequences of the violation of \cref{eq:qdigestProperty2}: having this operation violate \cref{eq:qdigestProperty2} would mean that $\qmerge$ is not an internal operation on the set of q-digest, and thus we lose many guarantees on the result. Particularly, the validity of this property is a premise of the \emph{q-digest size bound theorem}, which would not hold true anymore. This implies that we would have no theoretical bound to the size of the resulting object, and in turn no bound for the time needed to execute queries. This is especially a problem in the context the structure was designed for: distributed sensor networks, where having a known limit to both space and time complexity is crucial in order to achieve the best performance (and precision) possible.

Of the actual q-digest implementations tested \cite{qdigestcpp, qdigestjava}, we found out that the first one \cite{qdigestcpp} uses the original \textproc{Compress} algorithm (albeit calling it \textproc{compact}), however calling it only when the size of the structure grows above the $3k$ limit, which corresponds to the original size bound. The second one \cite{qdigestjava}, instead has two versions of the algorithm:
\begin{itemize}
	\item \textproc{compressUpward} is a faster version of compression, it only verifies \cref{eq:qdigestProperty2} along the path from the modified leaf to the root. The property can therefore be violated in some buckets.
	\item \textproc{compressDownward} is a complete and correct implementation of the q-digest compression, that guarantees that after its execution \cref{eq:qdigestProperty2} holds true for all buckets. This procedure is called only when the size of the structure has grown bigger than $3k$, or when invasive operations (such as a merge) happen.
\end{itemize}

This second implementation \cite{qdigestjava} only allows to insert a value (i.e., increase the count of a leaf) into the q-digest, and compresses it every time this happens. It is therefore not possible to construct a q-digest starting from a known set of frequencies, if not by repeatedly inserting values multiple times. This, however, causes the digest to tend to be less accurate than one constructed instantaneously from a set of frequencies, as more compression happens. Another consequence of this construction is that the order in which the values are inserted directly affects the shape of the resulting q-digest, and in turn its accuracy.

This kind of construction also does not verify the \emph{construction invariant} described in \cref{the:constructionInvariant}.

We now proceed by proposing two new compression algorithms and proving that they correctly restore the q-digest properties.

\subsection{Recursive Compression}
The first algorithm that we propose proceeds by recursively performing compression on subtrees of the q-digest.

We have seen that the compression error occurs when a bucket that has non-empty children is merged into its father. Consequently, an attempt at solving this issue could be going down a level whenever a parent bucket is compressed. An easy way to express this is:
\begin{enumerate}
	\item run a first compression on the two subtrees rooted in the children of the node;
	\item merge the two children in the node, if needed;
	\item if a merge has happened in the previous step, apply a second compression on the two subtrees.
\end{enumerate}
This is illustrated in \cref{alg:recursiveCompress}.

\setalgorithmicfont{\footnotesize}
\begin{algorithm}
\caption{Recursive Compress}
\label{alg:recursiveCompress}
\begin{algorithmic}[1]
	\Function{RecursiveCompress}{$Q$: q-digest, $b$: root}
		\If{$b$ is a leaf}
			\State \Return
		\EndIf
		\State \Call{RecursiveCompress}{$Q$, b.left}
		\State \Call{RecursiveCompress}{$Q$, b.right}
		\If{$b.cnt + \mathit{b.left.cnt} + b.right.cnt \leq \qlimit$}
			\State $b.cnt \gets b.cnt + \mathit{b.left.cnt} + b.right.cnt$
			\State $\mathit{b.left.cnt} \gets 0$
			\State $\mathit{b.right.cnt} \gets 0$
			\State \Call{RecursiveCompress}{$Q$, b.left}
			\State \Call{RecursiveCompress}{$Q$, b.right}
		\EndIf
	\EndFunction
\end{algorithmic}
\end{algorithm}

\subsubsection{Complexity analysis}\;\\
\indent To study the complexity of this algorithm, we note that it is a kind of recursive \emph{divide and conquer} algorithm, that in the worst case executes four recursive calls per step, each time with an input that is half of the size of the starting problem. Remembering that a q-digest tree is a binary partition tree of the domain of the q-digest, the size of the problem in this analysis is the width of the range partitioned by a tree or, equivalently, the number of leaves. The cost of the creation of the subproblems and of the combination of the results are $\bigo(1)$, and are therefore dominated by the cost of the recursive calls.

In these conditions, we can apply the first case of the \emph{master theorem for divide and conquer recurrences}~\cite{bentley1980general, cormen2001introduction} and obtain:
\begin{align*}
	T(\sigma) \leq 4T\left(\frac{\sigma}{2}\right) + \bigo(1)= \bigo(\sigma^2)
\end{align*}

Note that this cost is substantially higher than the one of the original procedure -- that is \compressCost --, and that it does not depend on $m$, the number of distinct values, but only on the size of the domain.

We finish by underlining that this analysis assumes that every call of the algorithm will result into four recursive calls, which is a highly unlikely, if at all possible, scenario, and that the real cost of this operation will probably be much lower. Furthermore, it would be possible and more efficient to execute a smarter version of the algorithm that runs not on the whole q-digest tree, but only on the buckets and the nodes that are modified during compression.

\subsubsection{Correctness}
We prove the correctness of the algorithm with the following theorem:
\begin{theorem}[\textproc{RecursiveCompress} Correctness]\label{the:recCompressionCorrectness}\;\\
	The \textproc{RecursiveCompress} algorithm correctly restores \cref{eq:qdigestProperty2} on all buckets of the q-digest, except, possibly, the root.
\end{theorem}
\begin{proof} First, we note that restoring \cref{eq:qdigestProperty2} on the root is not needed to satisfy \cref{def:qdigest}.\\
	We prove this theorem by induction:\\
	We have two base cases: in the case of a tree made up by only one node, the theorem is trivially true, as the only node is the root, and the theorem states that on it the property can either be satisfied or not.\\
	The second base case is when the algorithm is executed on a tree made up by three nodes: the root $a$ and its children $b$ and $c$. We have two subcases:
	\begin{itemize}
		\item $\Delta_a > \qlimit$, in which case it also holds true that \mbox{$\nabla_Q(b) = \nabla_Q(c) = \Delta_a > \qlimit$}.
		\item $\Delta_a \leq \qlimit$, in which case we update $a.cnt$ with $\Delta_a$. Now $b.cnt = c.cnt = 0$, which means they are no longer buckets, and the theorem is true.
	\end{itemize}

	Now for the inductive step, we need to consider the top seven nodes of the tree:

	\begin{figure}[H]
		\centering
		\inputtikz{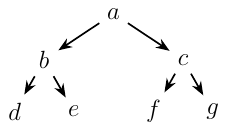}
	\end{figure}

	Initially, we execute a first compression on the two subtrees rooted in $b$ and $c$, which by inductive hypothesis are now compressed. We again have two cases:
	\begin{itemize}
		\item $\Delta_a = \nabla_Q(b) = \nabla_Q(c) > \qlimit$, in which case the theorem is true.
		\item $\Delta_a \leq \qlimit$, so we need to update the counter \mbox{$a.cnt \gets a_0.cnt + b_0.cnt + c_0.cnt$} and execute a second compression on the two subtrees. We denote with a subscript 0 the nodes as they were immediately after the first compression. Contrarily to the base case, we now still need to prove that the property holds for $b$ and $c$, as they might not be empty, and because they were roots in the recompression, the inductive hypothesis does not guarantee that the property holds. We prove it for $b$, and symmetrically it can be proved for $c$. This case itself is then split into two subcases:
		\begin{itemize}[leftmargin=1.5em]
			\item $b.cnt = 0$. In this case the property is restored, as empty nodes are not considered for compression.
			\item $b.cnt \neq 0$. In particular, in this case we have \mbox{$b.cnt = d_0.cnt + e_0.cnt$}, because the old $b_0.cnt$ has been moved up a level and is now part of $a.cnt$.
			
			Since $b.cnt \neq 0$ by hypothesis, this means that \mbox{$d_0.cnt \neq 0 \lor e_0.cnt \neq 0$} and consequently \mbox{$\Delta_{b_0} = b_0.cnt + d_0.cnt + e_0.cnt > \qlimit$} by inductive hypothesis.
			
			Since \mbox{$a.cnt = a_0.cnt + b_0.cnt + c_0.cnt$}, it follows that 
		\end{itemize}
	\end{itemize}
	{\small \begin{equation*}
		\nabla_Q(b) \geq \bunder{d_0.cnt + e_0.cnt}{$b.cnt$} + \bunder{a_0.cnt + b_0.cnt + c_0.cnt}{$a.cnt$} > \qlimit\,.
	\end{equation*}}

\end{proof}

\subsection{Iterative Compression}
Alternatively to the recursive approach, we describe an iterative algorithm which simply calls the original \textproc{Compress} algorithm multiple times, until a fixpoint is reached, and no bucket merge happens anymore. At this point we are guaranteed that \cref{eq:qdigestProperty2} holds true for all buckets, otherwise they would have been merged up during one of the iterations. This is illustrated in \cref{alg:iterativeCompress}.

\setalgorithmicfont{\footnotesize}
\begin{algorithm}
\caption{IterativeCompress}
\label{alg:iterativeCompress}
\begin{algorithmic}[1]
	\Function{IterativeCompress}{$Q$: q-digest}
		\Do
			\State \Call{Compress}{$Q$}
		\doWhile at least one merge has occurred in the last iteration
	\EndFunction
\end{algorithmic}
\end{algorithm}

\subsubsection{Complexity analysis}
To find a time complexity bound for this algorithm, we make the assumption that we are working on a structure that is the result of the \emph{sum of two q-digests}, which is perfectly reasonable if this procedure is used as part of a merge algorithm. In these conditions, the number of buckets of the structure is still $\bigo(k)$, as each of the source q-digests individually has $\bigo(k)$ buckets.

The analysis then is based on the observation that any leaf bucket can be merged up at most $\log \sigma$ times (the height of the q-digest tree). Therefore, by the \emph{pidgeonhole principle}, the maximum number of iterations will be $\bigo(k \log \sigma)$, if only one leaf goes up only one level during each iteration (which is an unrealistic scenario, but this simplifies the upper bound analysis). The cost of the algorithm is therefore $\bigo(k^2 \log^2 \sigma)$. This upper bound is quadratic with respect to the original compression, however note that in reality the upper bound will practically never be reached, as a single compression pass will not move only one leaf upwards. Most importantly, this algorithm is guaranteed to return a correctly compressed q-digest.

Note that if this algorithm is employed during the creation of a q-digest, its complexity drops to that of the original \textproc{Compress} algorithm, \linebreak\compressCost, with $m$ being the number of distinct values, as there will only be two iterations, the first one will compress the digest, and the second one will produce no change, determining the termination of the algorithm.

\subsubsection{Correctness}
\begin{theorem}[\textproc{IterativeCompress} Correctness]\label{the:iterCompressCorrectness}\;\\
	The \textproc{IterativeCompress} algorithm correctly restores \cref{eq:qdigestProperty2} on all buckets of the q-digest (except, possibly, the root).
\end{theorem}
\begin{proof} (Sketch)
	The correctness of this algorithm, as anticipated earlier, lies in the fact that no compressions have happened in the last iteration.

	Let us assume that, after the algorithm stopped running, $Q$ contains a bucket that does not satisfy \cref{eq:qdigestProperty2}.
	Since the algorithm stopped, we know that the last iteration did not compress any bucket.
	However, we know this to be impossible, as an iteration of \textproc{Compress} restores \cref{eq:qdigestProperty2} for all buckets and then possibly reintroduces some violation, but in any case if there is a bucket for which the property is not verified, then a merge \emph{will} happen. Therefore, we obtain a contradiction, and we know that the hypothesis that \cref{eq:qdigestProperty2} is violated is false. From this, the thesis follows, and that concludes the proof.
\renewcommand{\qedsymbol}{\Lightning}
\end{proof}

\section{Issues with the Original Size Bound Theorem}\label{sec:sizeBoundIssues}

\paragraph{Q-Digest Size Bound}\;\\
\indent An important result on the compression parameter $k$ is how it affects the size of a q-digest. We recall the original bound from Lemma~1 in \cite{shrivastava2004medians} with its proof, rewritten using our notation:

\begin{lemma*}[Size Bound]\;\\
	A q-digest $Q$ with compression parameter $k$ has size $\left|Q\right|$ less than $3k$.
\end{lemma*}
\begin{proof}\;\\
	$Q$ is a q-digest, therefore, as per \cref{def:qdigest}, all its nodes satisfy \cref{eq:qdigestProperty2}.
	This means that for each node $b$, we have
    \begin{equation*}
        \nabla_Q(b) > \frac{n}{k}
    \end{equation*}
	Summing this inequality on all buckets, we obtain
\begin{equation}    
	\sum_{v \in Q} \nabla_Q(v) > |Q| \frac{n}{k} \label[summation]{eq:summationP2Paper}
\end{equation}

Note that 
\begin{align}
	&\phantom{=}&&\sum_{v \in Q} \nabla_Q(v) \nonumber\\
	&\defeq&&\sum_{v \in Q}(v.cnt + v.parent.cnt + v.sibling.cnt) \nonumber\\
	&\leq&&3\sum_{v \in Q}v.cnt \label[inequality]{eq:upperBoundPaper}\\
	&=&&3n \nonumber
\end{align}

The \cref{eq:upperBoundPaper} is based on the observation that any bucket's count appears \emph{at most} three times, one in its own term, one as parent, and one as sibling. Therefore, it is possible to rearrange the terms of the summation in such a way that all three occurrences are in the same term. Doing this for all terms of the summation leads to this inequality.

Putting these together, we obtain
\begin{gather*}
	\left|Q\right| \frac{n}{k} < 3n \,,\\
    \text{and thus}\\
    \left|Q\right| < 3k
\end{gather*}

\renewcommand{\qedsymbol}{$\cancel{\square}$}
\end{proof}

This proof, however, has some problems:
\begin{itemize}
    \item \cref{eq:summationP2Paper} is not true, as \cref{eq:qdigestProperty2} is not necessarily verified for the root, so it cannot be included in the summation;
    \item the same summation, in the original article, despite referring to \cref{eq:qdigestProperty2}, uses $\frac{n}{k}$ instead of $\qlimit$, which is not a valid lower bound, because $\qlimit \leq \frac{n}{k}$;
    \item \cref{eq:upperBoundPaper} is not true, a counterexample can be seen in \cref{fig:qdigestSummationBigger3n}, where a valid q-digest has $\sum_{b \in Q} \nabla_Q(b) > 3n$. This is because the observation justifying this inequality states that any bucket appears at most three times in the summation, while in reality it might appear four times: one for its own term, one as sibling, and \emph{two} times as parent;
\end{itemize}
and, most importantly,
\begin{itemize}
    \item the lemma itself is not true, as it is possible to create a q-digest with a number of nodes greater than $3k$, as shown in \cref{fig:qdigestNumberNodesBigger3k}.
\end{itemize}

\begin{figure}[H]
    \centering
    \resizebox*{0.25\textwidth}{!}{
        \inputtikz{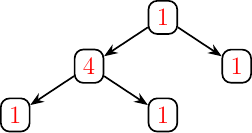}
    }
    \caption{Sample q-digest with $\sigma = 4, k = 2, n = 8$ for which the summation term is bigger than $3n$. Indeed:\\
    $\sum_{b \in Q} \nabla_Q(b) = 4 \times 6 + 1 = 25 \quad\nless\quad 24 = 3 \times 8 = 3n$.}
    \label{fig:qdigestSummationBigger3n}
\end{figure}

\begin{figure}[H]
    \centering
    \resizebox*{0.48\textwidth}{!}{
        \inputtikz{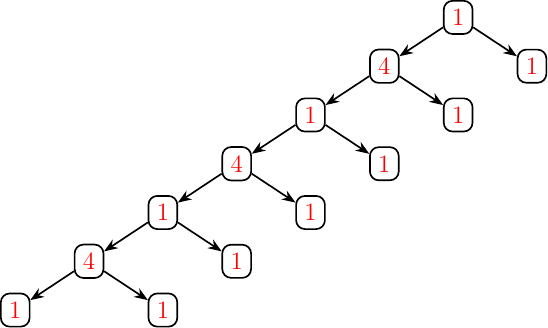}
    }
    \caption{Sample q-digest with $\sigma = 64, n = 22, k = 4$ that has more than $3k$ nodes, despite the fact that both properties in \cref{def:qdigest} are satisfied. Indeed: $\qlimit = \floor{\frac{22}{4}} = 5$, and $\left|Q\right| = 13 \quad\nless\quad 12 = 3k$.}
    \label{fig:qdigestNumberNodesBigger3k}
\end{figure}

Fortunately, it is still possible to find a similar bound that is still linear in the parameter $k$, albeit with a bigger factor.

\begin{theorem}[Q-Digest Size Bound]\label{the:sizeBound}\;\\
    A q-digest $Q$ with compression parameter $k$ has size $\left|Q\right|$ at most $4k + 1$.
\end{theorem}
\begin{proof}
    As in the previous proof, we use the fact that \cref{eq:qdigestProperty2} is satisfied for all nodes, \emph{except the root}. For each non-root node $b$, we have
    \begin{equation*}
        \nabla_Q(b) > \qlimit
        \quad\Rightarrow\quad
        \nabla_Q(b) \geq \frac{n}{k} \,.
    \end{equation*}
    Once again, we sum on all the nodes of the q-digest, \emph{except for the root}, obtaining
    \begin{equation}
        \hspace*{-1em}\sum_{v \in Q \setminus {Q.root}} \hspace*{-1.4em}\nabla_Q(v)
        \quad\geq\quad
        \left(\left|Q\right| - 1\right) \frac{n}{k}\,. \label{eq:newBoundInequality}
    \end{equation}

    Analogously to the previous proof, we find an upper bound for the summation:
    \begin{align}
        \hspace*{-1em}\sum_{v \in Q \setminus {Q.root}} \hspace*{-1.4em}\nabla_Q(v)
        \quad\leq\quad
        \sum_{v \in Q} \nabla_Q(v)
        \quad\leq\quad
        4n\,. \label[inequality]{eq:newBound4n}
    \end{align}
    Similarly to the original proof, this inequality follows from the observation that a node appears at most \emph{four} times in the summation, one in its term, one as sibling, and two times as parent, one for each child.

    We again proceed by putting (\ref{eq:newBoundInequality}) and (\ref{eq:newBound4n}) together, obtaining:
    \begin{equation*}
        \left(\left|Q\right| - 1\right) \frac{n}{k} \leq 4n \,, 
    \end{equation*}
    from which we obtain the bound
    \begin{equation*}
        \left|Q\right| \leq 4k + 1\,.
    \end{equation*}

\end{proof}
    
\begin{corollary}~\\
    If the \emph{construction invariant} holds true, then the size of a q-digest is at most $2k + 1$.
\end{corollary}
\begin{proof}    
    If the \emph{construction invariant} (\cref{eq:constructionInvariant}) holds true, then we can use $2n$ instead of $4n$ as upper bound in \cref{eq:newBound4n}, because no bucket can be the parent of other buckets, meaning that we can remove the parent counts from the summation, and each node can appear at most two times in the count. Hence, the limit becomes
    \begin{equation*}
        \left|Q\right| \leq 2k + 1 \,.
    \end{equation*}
\end{proof}

\section{Whole Digest Authentication}\label{sec:wda}
We now introduce a first technique to authenticate \mbox{q-digests}. This technique is lightweight in nature, and it is based on \emph{cryptographic hash functions}~\cite{menezes1997handbook}.\\
We introduce this approach for two reasons: first, to provide a baseline for other authentication techniques; second, because it is actually sufficiently adequate to use in many real use cases, due to its simplicity and to its low computational complexity.

To describe the scheme, we refer to the \emph{authenticated data structure model}~\cite{tamassia2003authenticated}. In this model, we have a \emph{user} $u$, a \emph{responder} $r$, and a \emph{source} $s$. The user $u$ wants to perform a query on a data structure $S$ held by $s$, but instead of sending the query directly to $s$, $u$ queries $r$, which holds a copy of $S$. In order to guarantee that $r$ has sent a correct response to the issued query, and that the data structure has not been tampered with, or that it has not been corrupted, some \emph{authentication information} about $S$ will need to be computed, which then will be used by $r$ to send a \emph{proof} of its response, which in turn will be used by the user to verify the correctness of the response.

\subsection{The Approach}
The approach proposed here, which we will call \emph{Whole Digest Authentication}, or \emph{WDA} in short, is very simple: a q-digest $Q$ is passed through a \emph{cryptographic hash function} $\mathcal{H}$, and the user $u$ stores $\mathcal{H}(Q)$. When $u$ needs to perform a query, the responder $r$ sends the whole digest $Q_r$. Then $u$ checks that $\mathcal{H}(Q) = \mathcal{H}(Q_r)$, and performs the query it needs to perform, or discards $Q_r$ if the hashes are not equal, recognizing there has been some corruption, whether it be intentional or not. Additional checks the user could perform are that both properties of \cref{def:qdigest} hold, and that the size of the structure fits between the bounds of \cref{the:sizeBound}. With these checks in place, we are essentially limiting the responses the user accepts to the set of valid q-digests, which should make a preimage attack on the hash even harder, as a malicious responder now has to ensure that the structure sent satisfies those properties, in addition to finding a value that produces the same hash.

\paragraph{Pros and Cons}\;\\
\indent This approach surely has the benefit of being very easy to implement and verify, but has the drawback that the user is required to have read permission on the whole structure, which could be undesirable in some setting where we want to keep the data private, and only disclose statistical data retrievable with queries. It indeed poses a problem in a \emph{non-public blockchain} setting, where users do not have read permissions on blocks. This issue will be addressed with the introduction of more advanced techniques in the next sections.

\section{KVC-Authenticated Queries}\label{sec:kvc-qa}
A different approach from \emph{WDA} is to provide authentication for the individual queries supported by the \mbox{q-digest}, allowing more granular proofs of authenticity and a more granular control on what data is exchanged between users and responders.\\
This approach makes use of \emph{Key Value Commitments} to authenticate q-digests, so that we can check membership of individual buckets on the authenticated structure. We call this approach \emph{Key Value Commitment Query Authentication}, or \emph{KVC-QA} in short.

As with \emph{WDA}, the q-digest is stored by the responder, and its authentication information is publicly accessible by users, provided for example by a trusted third party. In this instance, though, the authentication information is not simply the hash of the digest, but its KVC.

\subsection{Key Value Commitments}
There are several constructions for a Key-Value Commitment, from now on \emph{KVC}, such as those proposed in~\cite{agrawal2020kvac} and~\cite{campanelli2022zero}, and they are still being researched at the time of writing, with new constructions being proposed having interesting properties (e.g, the Z\kiwi proposed in~\cite{campanelli2022zero} is homomorphic, which could be of interest in our case).

For the purpose of this study, since we only need a limited subset of the operations supported by these commitments, we consider a generic, simplified, abstract KVC on the model of the \emph{KVaC} presented in~\cite{agrawal2020kvac}.

\begin{definition}\label{def:kvc}\;\\
A Key Value Commitment $KVC$ is a cryptographic primitive equipped with the following procedures:
\begin{align*}
	&Initialize() &&\rightarrow &&C\\
	&Insert(C, \langle key, value \rangle ) &&\rightarrow &&C'
\end{align*}
\end{definition}

$Initialize$ takes no argument and performs any operation needed to initialize the primitive. In a specific implementation this procedure could have additional input or output parameters (e.g, the security parameter $\lambda$ in a \emph{KVaC}), but since this is implementation-dependent, we do not concern ourselves with it.\\
We will denote with $K_{init}$ the time complexity of this procedure.

$Insert$ takes as argument a KVC obtained by \emph{initialization} or as a result of an(y number of) \emph{insertion(s)}. Similarly to the previous procedure, this one too could have more parameters depending on the specific implementation. This is the case with the update information returned by the \emph{KVaC} implementation, but since in our application we do not update proofs, we can disregard it.\\
We will denote with $K_{ins}$ the time complexity of this procedure.

\subsubsection{KVC Proofs}
We assume that the following property holds:
\begin{property}
Given two KVCs $C_1$ and $C_2$, that are respectively the result of a series of insertions of all the elements of the Key-Value Maps $M_1$ and $M_2$:\\
\begin{equation*}
	C_1 = C_2 \Rightarrow M_1 = M_2\; \text{ (with high probability)}
\end{equation*}

In other words, if two commitments are equal, then we expect the maps they represent to also be equal, with high probability.
\end{property}
Naturally, because of the nature of \emph{KVCs}, our proofs will be probabilistic, and we want the probability of generating two identical commitments representing different maps to be negligible. Depending on the actual implementation of the KVC, this could be justified by assumptions like the \emph{RSA assumption}~\cite{rivest2003rsa}.

\paragraph{Membership Query}\;\\
\indent We use the property just stated to prove membership of an element in a Key-Value Map authenticated by a KVC. The paradigm used is described in Algorithm~\ref{alg:membershipQuery}. Namely, to verify that the pair $(k,v)$ belongs to a KVM $M$, we use as \emph{proof} $P$ a KVC on the set $M \setminus \{(k,v)\}$, then we add the pair and check that it is equal to the KVC $C$ on the whole set $M$. This can then be generalized to execute multiple membership checks at once, by simply replacing the singleton set $\{(k,v)\}$ with the set of pairs we need to check, and then inserting them all.

\setalgorithmicfont{\footnotesize}
\begin{algorithm}
\caption{Membership Query}
\label{alg:membershipQuery}
\begin{algorithmic}[1]
\Function{Member}{$C$: KVC, $P$: proof, \mbox{$\langle$ $k$: key, $v$: value $\rangle$:} pair to test membership of}
	\State $P \gets$ \Call{Insert}{$P$, (k, v)}
	\State \Return $C = P$
\EndFunction
\end{algorithmic}
\end{algorithm}

In our specific application with q-digests, for non-membership queries it is equivalent to check whether a key is unset, or if the value associated to that key is $0$. This is because conceptually all nodes of the q-digest tree are buckets, but some of them have $count = 0$, and we say they do not belong to the q-digest.

\subsection{Authenticated Quantile Query}
To perform an \emph{authenticated quantile query}, reported in \cref{alg:authenticatedQuantileQueryComplexity}, first we need to sort the buckets according to a \emph{post-order traversal} on the nodes of the q-digest tree. We then proceed to compute the result of the query as if it were a regular, non-authenticated quantile query, by accumulating the counts of the buckets we encounter into a variable $count$. Once the current $count$ becomes greater than $qn$, for a certain bucket $b$, the result of the (non-authenticated) query will be $b.max$, the maximum element $y$ in the range $[x, y]$ represented by the bucket.\\
The difference with a non-authenticated query is that we need to compute a \emph{proof} for this query. This proof will consist of a commitment of all the buckets we have not visited during the query. We simply need to iterate over all the remaining buckets and insert them into an empty commitment.\\
We then return a tuple containing:
\begin{itemize}
	\item $b.max$: the result of the quantile query;
	\item $Q.sublist(0,b.index)$: the list of all the buckets that have been visited during the query, including $b$;
	\item $P$: the proof generated according to the procedure described above.
\end{itemize}

\setalgorithmicfont{\footnotesize}
\begin{algorithm}
\caption{Authenticated Quantile Query}
\label{alg:authenticatedQuantileQueryComplexity}
\begin{algorithmic}[1]
\Function{AQQ}{$Q$: q-digest, $q$: quantile}
	\State Sort $Q$ according to a post-order visit
	\CodeAnnotation{0.05em}{0}{\hspace*{0.2em}$K_{sort}$}{}
	\State $count \gets 0$
	\State $b \gets null$
	\State $i \gets 1$
	\CodeAnnotation{14.1em}{0.75}{$\bigo(1)$}{\\{}\\{}}
	\While{$b = null \land i \leq Q.length$}
		\State $count \gets count + Q[i].cnt$
		\If{$count \geq q \times n$}
			\State $b \gets Q[i]$
		\EndIf
		\State $i \gets i + 1$
	\EndWhile
	\CodeAnnotation{12em}{2.3}{$\bigo(k)$}{\\{}\\{}\\{}\\{}\\{}\\{}}
	\State Initialize commitment $P$
	\CodeAnnotation{5.8em}{0}{\hspace*{0.2em}$K_{init}$}{}
	\While{$i \leq Q.length$}
		\State $P \gets Insert(P, (Q[i].index, Q[i].cnt))$
		\State $i \gets i + 1$
	\EndWhile
	\CodeAnnotation{12em}{1.1}{$\bigo(k K_{ins})$}{\\{}\\{}\\{}}
	\State \Return $(b.max, sublist(Q, b), P)$
\EndFunction
\end{algorithmic}
\end{algorithm}

\subsubsection{Verification}
To verify the authenticity of the quantile query, we need to:
\begin{enumerate}
	\item check that the sum of the counts of all buckets in the sublist is greater than $qn$;
	\item check that the sum of the counts of all buckets in the sublist minus the last one is less than $qn$;
	\item prove membership of all the buckets in the sublist, and
	\item prove non-membership of all the buckets not in the sublist that would be counted before the bucket $b$.
\end{enumerate}
The first two points are needed to ensure that the algorithm stops at the correct bucket. If the user did not check the second point, the responder could send any value greater than or equal to the actual answer.
The last two points can be checked at the same time if we use a KVC that treats unset values as $0$ (as does the \emph{KVC-Inc increment-only commitment} described in~\cite{agrawal2020kvac}). The whole procedure is reported in \cref{alg:verifyQuantileQueryComplexity}.

Crucial point to note is that in \cref{alg:verifyQuantileQueryComplexity} the \emph{nodes} on which the user needs to iterate are not the buckets of the q-digest, but are all nodes of the q-digest tree. This is because we also need to prove non-membership of the empty nodes, for reasons explained in \cref{subsubs:aqqOmitLeft}.

\setalgorithmicfont{\footnotesize}
\begin{algorithm}
\caption{Quantile Query Verification}
\label{alg:verifyQuantileQueryComplexity}
\begin{algorithmic}[1]
\Function{QQV}{$m$: quantile query answer, $cb$: buckets counted during the query, $P$: proof, $C$: commitment of the whole digest}
\State $count \gets 0$
\CodeAnnotation{14em}{0}{\hspace*{0.2em}$\bigo(1)$}{}
\ForAll{$b$ before $cb[cb.length]$ in post-order}
	\State $P \gets Insert(P, (b.index, b.cnt))$
	\State $count \gets count + b.cnt$
\EndFor
\CodeAnnotation{15.2em}{1.15}{$\bigo(\sigma K_{ins})$}{\\{}\\{}\\{}}
\If{
	\hspace*{-0.7em}\raisebox{-1em}{$\Biggl($}\hspace*{-0.3em} \begin{varwidth}[t]{\linewidth}
		$(count \geq qn) \land$ \par
		$(count - cb[cb.length].cnt < qn) \land$ \par
		$P = C$
	\end{varwidth}\raisebox{-1em}{$\Biggr)$}\hspace*{-0.5em}
}
	\State \Return 1
\Else
	\State \Return 0
\EndIf
\CodeAnnotation{15.75em}{2.2}{$\bigo(1)$}{\\{}\\{}\\{}\\{}\\{}\\{}}
\EndFunction
\end{algorithmic}
\end{algorithm}

\begin{example}\label{ex:authenticatedQQ}\;\\
Suppose we want to authenticate a quantile query on the tree in \cref{fig:sampleqtree}, with $q=0.5$, $n = 15$, $\sigma = 8$.

\begin{figure}
	\resizebox{\columnwidth}{!}{
		\inputtikz{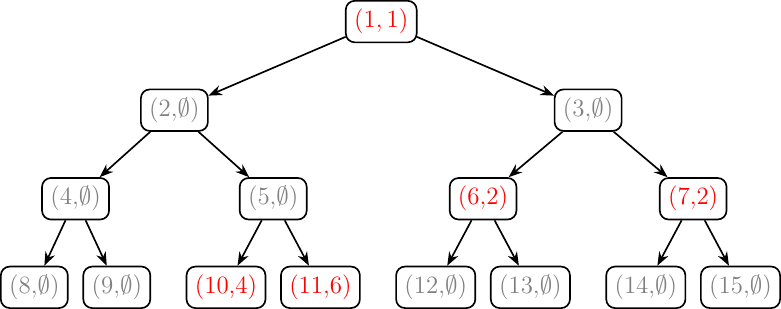}
	}
	\caption{Q-Digest ($n = 15$, $k = 5$, $\sigma = 8$) on which the query is being executed.}
	\label{fig:sampleqtree}
\end{figure}

\noindent Initially, the q-digest will be stored as a list of buckets:
\begin{equation*}
	Q = \{(1,1),(6,2),(7,2),(10,4),(11,6)\}
\end{equation*}

\paragraph{Authenticated query}~\\
As per \cref{alg:authenticatedQuantileQueryComplexity}, we proceed as follows:
\begin{enumerate}
	\item Sort $Q$ in post-order, which becomes\\
	$\{(10,4),(11,6),(6,2),(7,2),(1,1)\}$;
	\item initialize the variables and start executing the loop:
	\begin{enumerate}
		\item first, we check the first node in the list: $(10,4)$. We assign $count \gets 4$. Since this is less than $qn = 0.5 \times 15 = 7.5$, we continue the loop;
		\item we now add $(11,6)$, and $count \gets 4 + 6$. Since now $count > qn = 7.5$, we stop the loop here, remembering we stopped on the bucket $(11,6)$. We also save $b.max = 4$, as the range of the bucket with index $11$ is $[4,4]$. This will be the answer to the quantile query;
	\end{enumerate}
	\item initialize a commitment $P$ and insert all the remaining buckets: $\{(6,2),(7,2),(1,1)\}$;
	\item return $\left\langle b.max = 4, \{(10,4),(11,6)\}, P \right\rangle$.
\end{enumerate}

\paragraph{Query verification}
In order to verify the query, we need the following information:
\begin{itemize}
\item The triple returned by the query,
\item $n$, which is the count of all items stored in the digest (i.e., the sum of all counts),
\item $C$, the commitment of the whole digest.
\end{itemize}
As per \cref{alg:verifyQuantileQueryComplexity}, these are the steps to follow:
\begin{enumerate}
\item Initialize $count$ to $0$,
\item Insert all the nodes of the full binary subtree that would be visited before $b$ into the proof. These nodes in this case would be:
\begin{equation*}
	\{(8,\emptyset),(9,\emptyset),(4,\emptyset),(10,4),(11,6)\}
\end{equation*}
\item Sum the count of all nodes counted in the previous step: 
\begin{equation*}
	count \gets 0 + 0 + 0 + 4 + 6 = 10
\end{equation*}
\item Check that
\begin{itemize}
\item $count = 10 > 7.5 = qn$
\item $count - b.max = 10 - 4 = 6 < 7.5 = qn$
\item $P = C$
\end{itemize}
\item If and only if all the checks succeed, then the query is verified.
\end{enumerate}

\end{example}

\subsubsection{Complexity Analysis}
The time complexity of the algorithms described above naturally depend on the $K_{init}$ and $K_{inc}$ of the specific $KVC$ taken into consideration, as they include calls to the procedures \textproc{Initialize} and \textproc{Insert}.

Looking at the annotations in \cref{alg:authenticatedQuantileQueryComplexity} we note that the cost is $\bigo(K_{sort} + k K_{ins} + K_{init})$, where $K_{sort}$ is $0$ if we assume that $Q$ is already sorted.

For the verification part, looking at the annotations in \cref{alg:verifyQuantileQueryComplexity}, we can see that in the worst case it costs $\bigo(\sigma K_{ins})$. The real cost of performing the authenticated query and verifying it depends on how early the query ends, specifically on how many buckets the algorithm has to visit before returning the answer. In fact, the responder will insert all the buckets that \emph{have not} been counted in the query, while the user will have to insert in a KVC all the buckets that \emph{have} been counted, plus the empty nodes that appear before the bucket where the algorithm stopped. These empty nodes are the ones that make up the majority of the cost, as they are not linear in $k$, but are actually linear in $\sigma$, as they are nodes of the q-digest tree.

At the extremes, if $q=0$ the responder will have to \textproc{Insert} all buckets (except the first one), while the user only has to \textproc{Insert} that first one, leading to costs of query and authentication respectively of $\bigo(k K_{ins} + K_{init})$ and $\bigo(K_{ins})$. This is the best case scenario for the user and the worst case scenario for the responder. Symmetrically, if $q=1$, then the responder will just have to send an empty commitment as proof, while the user will have to \textproc{Insert} all the nodes of the full q-digest tree associated with $Q$. In this case, the costs become $\bigo(k + K_{init})$ for the query, and $\bigo(\sigma K_{ins})$ for the verification part. This is the best case scenario for the responder and worst case scenario for the user.

\subsubsection{Breaking the Protocol by Omitting buckets from the left}\label{subsubs:aqqOmitLeft}
Suppose that \cref{alg:verifyQuantileQueryComplexity} did not insert empty nodes in the commitment.
In this case, a malicious responder could avoid counting buckets from the left of the bucket $b$ (where the algorithm stops), resulting in a wrong answer to the quantile query. For example, considering the q-digest we used in \cref{ex:authenticatedQQ}, the responder could execute the query by only counting nodes in the subset $\{(11,6),(6,2),(7,2),(1,1)\}$, omitting the first bucket $(10,4)$. The reported answer to the query would then be $6$, as the computation would terminate on the bucket $(6,2)$, which has range $[5,6]$.

Now, if this were the only way the responder broke the protocol, then the user who is verifying the query would reject it because the bucket $(10,4)$ is not present in the commitment it has computed, resulting in a commitment that is different from $C$.

However, the malicious responder could bypass this problem by simply adding the bucket $(10,4)$ to the proof $P$, as if it were on the right-hand side of the node $b$. In this way, the user would compute the correct commitment, resulting in an incorrect query being recognized as valid.

The method we adopt to avoid this kind of attack is to insert all the nodes of the complete binary tree to the commitment during the verification process, including empty ones. By doing this, a non-membership check of the node with index $10$ (or equivalently a membership check of $(10,0)$) will fail, resulting in the query correctly being rejected.

This, however, comes at a very high computational cost, as it greatly increases the verification complexity from $\bigo(k K_{ins})$ to $\bigo(\sigma K_{ins})$.

\subsection{Feasibility of the Approach}\label{sec:kvcqaFeasibility}
The attack just described in \cref{subsubs:aqqOmitLeft} reveals a core issue with the proposed approach: not only the user needs to verify membership for the non-empty buckets, they also need to verify the empty ones, to guarantee that nodes that the user thinks are empty have not been added to the proof by the responder. This essentially means that we cannot use usual techniques to accelerate verification (e.g., by using Merkle trees on subtrees of the q-digest tree), \emph{in a way that is agnostic of the data included in the digest itself}. The issue is structural, as it is a consequence of the ``fingerprinting'' nature of the KVC: the user is not able to see what elements are included in a KVC (which would be like ``reversing'' it in some way), otherwise they could notice that the responder has not been honest with the answer.

As this authentication method is ridden with this problem, the range of applications where this approach is feasible is reduced, as it greatly depends on the computational capabilities of the users.

\section{Comparison Between WDA and KVC-QA}\label{sec:wdavskvcqa}
\subsection{Space and Time Complexity}
The two methods of \emph{Whole Digest Authentication} and \emph{Key Value Commitment Query Authentication} both have some advantages and disadvantages.

First, it is important to underline that they are two fundamentally different ways of authenticating the structure: one (\emph{WDA}) authenticates the digest as raw data, as a black box, and then it is responsibility of the user knowing how to perform queries; the other one \mbox{(\emph{KVC-QA})}, instead, is in some way query-aware, in the sense that it allows to authenticate the query process itself. Using KVCs is of course also more granular than using a simple hash, in the sense that it allows to verify that one element (in our case, one bucket) belongs to the set, instead of having to verify the whole set and then check for membership.

In general, WDA requires less computational time: the responder does not need to perform any calculation, and the user just needs to calculate $\mathcal{H}(Q)$ and compare it with the expected hash, while, as we discussed, KVC-QA requires $\bigo(k K_{ins})$ for the query, and $\bigo(\sigma K_{ins})$ for verification, and these costs could make it prohibitive in many applications, even by choosing a KVC with \mbox{$K_{ins} = \bigo(1)$}. As far as space is concerned, WDA requires the user to store the whole digest, while KVC-QA only requires the exchange of a subset of buckets and of the proof, so in principle it would need less space, if the KVC chosen has constant space requirements.

\subsection{Privacy}\label{subs:privacy}
The aspect where KVC-QA wins over WDA is on privacy: WDA requires the responder to send the whole digest to the user, which might be undesirable in a scenario where we do not want to disclose the whole digest, but just want to give the user the least possible information to answer the query. This issue will be later discussed in \cref{sec:zkqd}.
By contrast, KVC-QA are somewhat better in this regard, as authenticating a query does not necessarily require the user to know the whole content of the q-digest. However, a user who wants to obtain all the nodes might still do it, by simply performing a quantile query with $q=1$ as in order to authenticate such query, the responder has to send all the buckets in the digest to the user.

\paragraph{Lowering Precision to Enforce Privacy}\;\\
\indent Suppose we are in a scenario where we want to control the precision of the queries performed by users, for example by allowing a group of users to perform exact queries, a second group to perform approximate queries, and finally a third group to perform approximate queries with a lower precision than that of the second group, and maybe we do not want to allow users in this last group to gain information on high frequency values in the distribution (e.g., leaves with a count high enough that they have not been compressed). We could exploit the approximating nature of q-digests to obtain these goals.

\subsubsection{\texorpdfstring{$k$}{k} as Privacy Parameter}\label{subs:privacyk}
The first method we present to enforce some privacy on the structure is to simply -- and quite naturally, due to the nature of the structure -- lower the value of the parameter $k$, in order to increase the compression of the q-digest and disclose less precise information to the user.

This approach would be used roughly in this way:
\begin{enumerate}
    \item choose a number $l$ of \emph{privilege levels} $p$, where $p_1$ is the most privileged level and $p_l$ is the least privileged one, and in general $p_1 > p_2 > ... > p_l$;
    \item assign an increasing monotonic mapping between privilege levels and different values of $k$, such that $k_{p_1} > k_{p_2} > ... > k_{p_l}$;
    \item when creating a q-digest, for each $p_j$ compute the corresponding $Q_{p_j}$, with the appropriate $k_{p_j}$ value;
    \item store each digest $Q_{p}$ separately, and publish authentication information for each one ($\mathcal{H}(Q_{p})$ or KVC($Q_{p}$), depending on the method used);
    \item when user $u_{p_j}$ (with privilege level $p_j$) sends a query, use the corresponding digest $Q_{p_j}$ to answer the query.
\end{enumerate}

The pros of this system are that as we said, it is very natural due to the properties of the structure, and that we already have bounds on the query error expressed in terms of $k$, so we can predict upper bounds for each privilege level. On the other hand, though, we cannot enforce lower bounds for the error, meaning that for some query the users could still access more precise data than what we would like. Furthermore, by only changing the value of $k$, we cannot ensure that a leaf will be necessarily merged into a parent node, the only case in which we can ensure a q-digest has no leaves (without knowing its distribution a priori) is the degenerate case with $k=1$, where the q-digest collapses into one single node, the root.

\subsubsection{Coarse-Grained Q-Digests}
To address the shortcomings of using $k$ to restrict the information a user receives, we propose another approach we call \emph{coarse-grained q-digests}. Usually, a q-digest is built on a binary partition of the domain, where for each level of the binary tree we divide the range in smaller ranges, up until the leaves, which correspond to unit ranges. Of course, we do that because we want the digest to be as precise as possible. However, if we wish to prevent the users from accessing unit-range information, that is, information about the frequency of a single element of the domain, we can remove this information from the q-digest altogether, by stopping the binary partition earlier, effectively building a q-digest whose leaf buckets do no represent unit ranges, but bigger ones instead. The construction then would proceed, as one can imagine, by distributing the initial values in leaves with the appropriate range.
An analogous way of seeing this q-digest is as a regular one, except that the last $l$ levels of the tree have been forcefully compressed upwards, disregarding the two properties in Definition~\ref{def:qdigest}, then all the nodes in those levels are removed from the tree. Lastly, a regular compression takes place, as is for the usual creation of a q-digest.

What is interesting regarding this approach is that it is in some way complementary to the previously discussed one, in that it allows to force a lower bound on the query error. As such, it is possible to combine the two methods to have control on both lower and upper bounds at the same time, allowing for more control on how much information the users are able to access.

\subsection{Cumulative Digest}\label{sec:cumulativeDigests}
Let us now consider a scenario where we have a sequence of q-digests built from data coming at different moments in time. If users want to query the whole sequence of digests, the naive alternative could be performing queries on every single digest, but with the passing of time, this can become prohibitive. A more interesting -- and less expensive -- alternative involves building a q-digest $Q^c$ representing cumulative information for all digests up to this point.\\
Let $Q_i$ be a q-digest representing the data generated at time $i$, we can recursively define a cumulative q-digest in this way:
$$
    Q^c_i = Q^c_{i-1} \qmerge Q_i \,.
$$
Naturally, for the first block we define
$$
    Q^c_0 = Q_0\,,
$$
as there is no previous block whose data we need to consider.

However, the \emph{cumulative digest} we have just described presents in practice a major flaw: the number of values stored ($n$) keeps growing unbounded with every digest added, while the compression parameter $k$ has to be fixed at the creation of the digest. Since the merge of two buckets is determined by a comparison with $\qlimit$, this means that at some point the data added will be excessively compressed, to the point of losing any meaningful value.

This phenomenon is particularly significant for some distributions of the values stored. Particularly interesting is the case where initially the values inserted in the q-digest are mostly in the lower half of the range, which causes the right side of the tree to be compressed upwards, and then at a later time the distribution shifts so that the values are mostly found in the upper half of the universe range. At this time, the new values inserted from the q-digest $Q_i$ will not be able to create ``meaningful'' buckets (closer to the leaves) in the cumulative digest $Q^c$, assuming that $n_{Q_i} < \floor{\frac{n_{Q^c}}{k}}$, and will be compressed upwards. At the same time, the increase in $n_{Q^c}$ will eventually cause the left side of the tree to be compressed as well. In the worst case, this could lead to a q-digest where the only existing buckets are a handful of nodes close to the root, and any new q-digest merged into it would result in the new data being merged into these buckets or one of their direct children. Of course, a q-digest built on a very large dataset, that contains only a few buckets representing big ranges, cannot provide adequate accuracy on most queries.

\subsubsection{Partial Cumulative Digest}
A solution to the aforementioned issue is to provide cumulative digests on smaller, limited ranges of digests, essentially creating a kind of skip list of cumulative digests. This way, for a cumulative digest $Q^c$ over $w$ different digests, $n_{Q^c}$ does not grow unbounded, but we can expect it to be roughly $w \times n_{Q_i}$ (assuming the number of values in each $Q_i$ are similar and do not deviate too much), and this alleviates the problem. We could choose a value for $k$ accordingly to avoid too much compression. Another idea could be building this ``skip list'' by building cumulative q-digests over a sliding window of single q-digests.

\section{Further developments}\label{sec:furtherDevelopments}

\subsection{Compression Parameter}
\subsubsection{Recompressing Quantile Digests}\label{subs:digestRecompression}
For now, all the operations on q-digests have left the parameter $k$ unchanged. Being able to change it after the digest's creation could open the way to some interesting uses of the data structure. For instance, it would be possible to use the compression-based privacy strategy described in \cref{subs:privacyk} without needing to store multiple versions of the same q-digest, but simply building them on-demand. Unfortunately, even in this case, the authentication information still has to be calculated prior and kept in a trusted store. Of course, this strategy only works if the derivation of the $Q_{p}$s from $Q$ is deterministic, otherwise the responder would obtain a digest that cannot be authenticated with the information accessible by the user.

Unfortunately, it is not immediately clear how feasible recompressing a q-digest is, and how its properties would be affected. Surely, recompressing it with a value $k' > k$ would be more problematic, as that would mean that $\floor{\frac{n}{k'}} \leq \qlimit$, which could invalidate \cref{eq:qdigestProperty1}, and contrarily to \cref{eq:qdigestProperty2}, we do not have a procedure to restore it on the structure. Furthermore, increasing the value of $k$ intuitively means lowering the level of compression, and since the compression is lossy, obviously we cannot obtain an exact ``decompression''.

More realistic is the case where $k' < k$, as we do not have this issue, and it is possible to simply use the \textproc{Compress} procedure to restore Property~\ref{eq:qdigestProperty2} if needed.

\subsubsection{Merging Q-Digests with Different Values of \texorpdfstring{$k$}{k}}
It could be interesting and useful to study what happens if one tries to merge two q-digests having a different value for the compression parameter $k$. If the possibility discussed above of recompressing digests turns out to be feasible, then a simple method to achieve this kind of merge would directly follow: it would be sufficient to recompress one of the two q-digests and then merge them.

\subsection{Homomorphic KVCs for Q-Digests}
\paragraph{Homomorphic Cryptography}\;\\
\indent \emph{Homomorphic cryptography}~\cite{henry2008theory} techniques allow operations to be performed on encrypted data directly, without needing to go through the process of decrypting the data and then encrypting it again after performing said operation. What this means for us, is that by using a \emph{Homomorphic Key-Value Commitment} to authenticate q-digests (such as the one in \cite{campanelli2022zero}), it could technically be possible to authenticate some composition of multiple digests, namely sum and merge of two q-digests, by only using the KVCs of the two starting digests.
In formulae, what we would like to obtain is:
\begin{align*}
	KVC(Q_1 \qadd_{Q} Q_2) &= KVC(Q_1) +_{KVC} KVC(Q_2)\\
	KVC(Q_1 \qmerge_{Q} Q_2) &= KVC(Q_1) \qmerge_{KVC} KVC(Q_2) \,.
\end{align*}

The first one, the sum of two q-digests, seems reasonable, as by \cref{def:qsum} the sum of two q-digests is an operation very similar to a set union, with the caveat that counts have to be added together for buckets with the same index, but this should be a fairly basic operation to implement in a homomorphic KVC.

The q-digest merge, on the other hand, requires some non-trivial operations, as it includes a compression, and it is not clear if such operation is possible on a homomorphic KVC.

\subsubsection{Reducing the Cost of KVC Verification}
We have discussed the reason why an authenticated quantile query using KVC-QA has a high worst case verification cost in \cref{sec:kvcqaFeasibility}.
This high computational cost could be attenuated by an arbitrary constant, using homomorphic KVCs, by storing commitments of partial subtrees of the q-digest tree. For example, one could store both a commitment of the whole digest (needed for verification), and a commitment of the subtree rooted in the left child of the root. This way, during verification, the user will never need to insert more than half of the nodes of the whole tree:
\begin{itemize}
	\item if bucket $b$ (where the algorithm stops) is in the subtree rooted in the left child of the root, then the verification algorithm proceeds as usual, and the maximum number of insertions needed is $\sigma$;
	\item if bucket $b$ is in the subtree rooted in the right child of the root, or the root itself, then the verification algorithm starts by summing the proof it obtained as response from the responder to the commitment of the subtree rooted in the left child of the root. By doing this, we effectively eliminate the need of adding all the nodes that are already covered by the subtree commitment.
\end{itemize}
Instead of storing the commitment for only one subtree, we can store an arbitrary number of them, and, during verification, use the one that minimizes the number of nodes that still need to be added. Choosing which one to use can be done efficiently by using the level order index of the node where the algorithm stopped. Of course this method requires a number of new commitments to be saved, which increases the memory footprint.

Another method is to store commitments of partitions of the whole q-digest tree, and to sum them during verification. This method though, requires a little more time to choose the commitments and sum them.

This approach can also work if we just have a KVC where we can deterministically obtain a sum of two (partial) commitments.

\subsubsection{Authenticating the Merge of two Q-Digests}
If we had a KVC that is homomorphic with respect to the merge operation on q-digests, it would make some applications more appealing, such as the cumulative digests described in \cref{sec:cumulativeDigests}. Specifically, a cumulative digest could be dynamically computed by the responder, and the user would only need the authentication information for the q-digests used in order to be able to authenticate it. Without the possibility of merge authentication, the authentication information for every cumulative digest has to be computed in advance and then stored somewhere accessible to the user.

Because of what we have seen during the discussion of the original \textproc{Compress} procedure in \cref{sec:compressionIssues}, we could decide to build the cumulative digest using the sum of two q-digest, instead of their merge, because in the end the only property that would not hold is the size bound of \cref{the:sizeBound}. This could be fine when dealing with a small number of q-digests, but when the number of included digests increases, the size of the structure could grow too big.

\subsubsection{Authenticating a Q-Digest Compression}
Should the compression of a q-digest turn out to be authenticable by using the authentication information of the digest prior to its compression, and if the recompression of a q-digest with a different $k$ also turns out to be possible, a useful consequence is that the scheme presented in \cref{subs:digestRecompression} could be further improved by eliminating the need of computing the authentication information for each $Q_{p}$, and in turn eliminating the need of computing the $Q_{p}$s themselves. This is because the recompressed $Q_{p}$ could be computed on-demand by recompressing $Q$ with the appropriate value for $k$, and the authentication information needed is the same used for $Q$ itself, meaning there is no need to precompute it.

\subsection{Reducing Information Complexity}\label{sec:zkqd}
As we have discussed throughout this work, different authentication methods reveal more or less information related to the structure of the q-digest. Some methods, like WDA, completely disclose the whole structure, while others, like KVC-QA, only reveal a subset of nodes to the user. The strategies discussed in \cref{subs:privacy} act on a different level, not changing the information that is disclosed during a query, but rather restricting the information that will be inserted in the q-digest at its creation.\\
A true ``zero-knowledge'' q-digest, where the answer to quantile queries reveals no information on the structure itself, seems contradictory, as the answer itself needs to convey some information on the underlying distribution. Due to the nature of quantile queries, in fact, with a sufficient number of them it is possible to completely extract the distribution (since it is discrete), and in turn to construct another q-digest that is identical to the source.

Despite that, it might still be desirable to devise an authentication method that discloses the least amount of information possible. Ideally, only the answer to the queries should be passed to the user, with no information on the internal structure.

\section*{Conclusion}
With this paper, we have described two authentication techniques, each with its advantages and disadvantages. We have also outlined some existing issues in the data structure, and presented solutions.

We believe that these results will be useful solutions to the initial problem of providing an authenticated, but compact, data structure. We also believe that the work done in this paper will provide a good groundwork and important insight that will be helpful to improve and expand on these techniques.

\printbibliography

\end{document}